\documentclass[11pt]{article}
\usepackage[latin1]{inputenc}
\usepackage{cite}
\usepackage{xcolor}
\usepackage{amssymb}
\usepackage{amsmath}
\usepackage{amsthm}
\usepackage{algpseudocode}
\usepackage{algorithm}
\usepackage{comment}
\newtheorem{theorem}{Theorem}

\newtheorem{lemma}[theorem]{Lemma}
\newtheorem{corollary}[theorem]{Corollary}

\title{Geometric decoding of subspace codes with explicit Schubert calculus applied to spread codes}
\author{Klara Stokes}
\begin{document}
\maketitle

\begin{abstract} 
This article is about a decoding algorithm for error-correcting subspace codes. 
A version of this algorithm was previously described by Rosenthal, Silberstein and Trautmann \cite{Rosenthal2}.
The decoding algorithm requires the code to be defined as the intersection of the Pl\"ucker embedding of the Grassmannian and an algebraic variety. 
We call such codes \emph{geometric subspace codes}.   
Complexity is substantially improved compared to \cite{Rosenthal2} and connections to finite geometry are given. 
The decoding algorithm is applied to Desarguesian spread codes, which are known to be defined as the intersection of the Pl\"ucker embedding of the Grassmannian with a linear space. 
\end{abstract}

\section{Introduction}
Consider the vector space $V=V(n+1,\mathbb{F})$ of dimension $n+1$ over a field $\mathbb{F}$. 
The Grassmannian $G_{\mathbb{F}}(k+1,n+1)$ is the set of subspaces of dimension $k+1$ of $V$. 
The projective geometry of dimension $n$ over $\mathbb{F}$ is the collection of Grassmannians $PG(n,\mathbb{F})=\{G_{\mathbb{F}}(0,n+1),\dots,G_{\mathbb{F}}(n+1,n+1)\}$. 
In particular, a subspace of $V(n+1,\mathbb{F})$ of dimension $k+1$ corresponds to a projective subspace of dimension $k$ of the projective geometry $PG(n,\mathbb{F})$. If $U$ is a vector subspace of $V(n+1,\mathbb{F})$ of dimension $k+1$, then we  denote by $\mathbb{P}(U)$ the projective subspace of dimension $k$ associated to $U$, and we say that $\mathbb{P}(U)$ is the projectivization of $U$. 
For example, the projectivization of the proper subspaces of $V(4,\mathbb{F})$ of dimension $1$, $2$ and $3$ are the points, the lines and the planes of $PG(3,\mathbb{F})$ respectively. 

A \emph{subspace code} in $V$ is a set of subspaces of $V$, therefore a set of projective subspaces of $PG(n,\mathbb{F})$. 
If all subspaces of the code have the same dimension $k+1$, then the code is contained in the Grassmannian $G_{\mathbb{F}}(k+1,n+1)$. Such codes are sometimes called constant-dimension codes or Grassmannian codes. 
Subspace codes are used for example in network coding \cite{subspacecodes,networkcoding}. 
Error-correction for subspace codes in network coding was introduced in~\cite{errorcorrection}.  
Just as in classical coding theory, error-correction in subspace codes requires the codewords to be taken well-separated according to some distance. 
A natural distance between two subspaces $A$ and $B$ in a vector space $V$ is $d(A,B)=\dim(A+B)-\dim(A\cap B)=\dim(A)+\dim(B)-2\dim(A\cap B)$. 
Illustrating with an example, for this distance, in $PG(3,\mathbb{F})$ the maximal possible distance between two codewords is 4, and this bound is attained by a set of lines with pairwise empty intersection. 
It is a non-trivial fact that it is possible to partition the set of points in  $PG(3,\mathbb{F})$  with a set of non-intersecting lines. 
Such a set of lines is called a \emph{spread} of lines in $PG(3,\mathbb{F})$. 

In general, a $t$-spread in $PG(n,\mathbb{F})$ is defined as a set of subspaces of projective dimension $t$ that partitions the point set, so a spread is simply a 1-spread. 
There is a $t$-spread in $PG(n,\mathbb{F})$ if and only if $(t+1)|(n+1)$, see \cite{Andre}. 
Spreads were first used as subspace codes for network coding in \cite{MaGoRo}. 
There are several decoding algorithms for spread codes, see \cite{MaGoRo, spreadcodes2,trautmann} and \cite{errorcorrection, Silva}. Partial spreads have also been used \cite{gorla}. 

Schubert calculus was proposed for error-correction of subspace codes in \cite{Rosenthal}, and applied to give an algorithm for list-decoding using explicit Schubert calculus in \cite{Rosenthal2}.  
In this article we use explicit Schubert calculus to decode geometric subspace codes, that is, subspace codes which are algebraic varieties in the Pl\"ucker embedding of the Grassmannian. We give substantial improvements to the methods described in \cite{Rosenthal, Rosenthal2,trautmann} and show that the complexity of this algorithm can be reduced  to the order of complexity of solving a system of linear equations in $\bigwedge^{k+1}V(n+1,\mathbb{F})$. In \cite{Rosenthal2}, the algorithms had the order of complexity corresponding to solving a system of quadratic equations in $\bigwedge^{k+1}V(n+1,\mathbb{F})$, which is at least exponential. 
We also show that for codes correcting one error the algorithm does not require passing to Pl\"ucker coordinates, so in that case the complexity of the algorithm has the order of complexity of solving a system of linear equations in $V(n+1,\mathbb{F})$. 
Finally, we apply the algorithm to the important case when the subspace code is a Desarguesian $t$-spread. 

Section \ref{section:2} gives a detailed description of the algorithm for a Desarguesian line spread code in $PG(3,\mathbb{F})$. This code corrects one error.  
The description of this algorithm is very thorough and should require less background in algebra, compared to the subsequent sections.  
In Section \ref{section:3} we give the details required for explicit Schubert calculus for decoding, define three distinct versions of the decoding algorithm for geometric subspace codes and calculate their complexity. 
In Section \ref{section:4} we apply the decoding algorithm to $t$-spreads in $PG(2t+1,\mathbb{F})$.

\section{A geometric decoder for a Desarguesian spread of lines in $PG(3,\mathbb{F})$}
\label{section:2}
The Grassmannian of lines in $PG(3,\mathbb{F})$ is the smallest interesting Grassmannian. It was first discovered by Pl\"ucker that a line in this Grassmannian can be given coordinates. 

Let $(a_0:a_1:a_2:a_3)$ and $(b_0:b_1:b_2:b_3)$ be the projective coordinates of two distinct points on a line $\ell$. The primary Pl\"ucker coordinates of $\ell$ are then $$\begin{array}{l}\left(q_0:q_1:q_2:q_3:q_4:q_5\right)=\\\\\left(\left|\begin{array}{cc}a_0&a_1\\b_0&b_1\end{array}\right|:\left|\begin{array}{cc}a_0&a_2\\b_0&b_2\end{array}\right|:\left|\begin{array}{cc}a_0&a_3\\b_0&b_3\end{array}\right|:\left|\begin{array}{cc}a_1&a_2\\b_1&b_2\end{array}\right|:\left|\begin{array}{cc}a_1&a_3\\b_1&b_3\end{array}\right|:\left|\begin{array}{cc}a_2&a_3\\b_2&b_3\end{array}\right|\right)\end{array}.$$ The primary Pl\"ucker coordinates of a line represent a point in $PG(5,\mathbb{F})$. 

Let $a_0^*X_0+a_1^*X_1+a_2^*X_2+a_3^*X_3=0$ and $b_0^*X_0+b_1^*X_1+b_2^*X_2+b_3^*X_3=0$ be the projective equations of two distinct planes intersecting in a line $\ell$.  
The dual Pl\"ucker coordinates of $\ell$ are then 
$$\begin{array}{l}
\left(q_0^*:q_1^*:q_2^*:q_3^*:q_4^*:q_5^*\right)=\\\\\left(\left|\begin{array}{cc}a_0^*&a_1^*\\b_0^*&b_1^*\end{array}\right|:\left|\begin{array}{cc}a_0^*&a_2^*\\b_0^*&b_2^*\end{array}\right|:\left|\begin{array}{cc}a_0^*&a_3^*\\b_0^*&b_3^*\end{array}\right|:\left|\begin{array}{cc}a_1^*&a_2^*\\b_1^*&b_2^*\end{array}\right|:\left|\begin{array}{cc}a_1^*&a_3^*\\b_1^*&b_3^*\end{array}\right|:\left|\begin{array}{cc}a_2^*&a_3^*\\b_2^*&b_3^*\end{array}\right|\right)\end{array}.$$ 
The dual Pl\"ucker coordinates of a line also represent a point in $PG(5,\mathbb{F})$. 

\begin{lemma}
\label{lemma:1}
\cite{Hodge} The primary and the dual Pl\"ucker coordinates of a given line $\ell$ are related by $\left(q_0:q_1:q_2:q_3:q_4:q_5\right)=\left(q_5^*:-q_4^*:q_3^*:q_2^*:-q_1^*:q_0^*\right)$.  
\end{lemma}

The image of the Grassmannian $G_{\mathbb{F}}(2,4)$ under either one of the primary and the dual Pl\"ucker coordinates is an algebraic variety in $PG(5,q)$ defined by the equation $x_0x_5-x_1x_4+x_2x_3=0$. 
This is a hyperbolic quadric  $Q=Q^+(5,\mathbb{F})$, in this context known as the Klein quadric. It contains points, lines and planes, which represent $PG(3,\mathbb{F})$ in a complete incidence preserving correspondence, which is described in what follows. 

The planes contained in $Q$ can be partitioned into two equivalence classes $A$ and $\Omega$ defined by the relation $\pi\sim \tilde{\pi}$ if $\pi\cap\tilde{\pi}$ is a point or if $\pi=\tilde{\pi}$. 
The Pl\"ucker coordinates of the set of lines through a point $p$ in $PG(3,\mathbb{F})$ are the points of a plane contained in $Q$, and each of the planes in the equivalence class of planes $A$ corresponds to a point in $PG(3,\mathbb{F})$ in this way. 
Also, the Pl\"ucker coordinates of the set of lines contained in a plane $P$ in $PG(3,\mathbb{F})$ are the points of a plane, and each of the planes in the equivalence class of planes $\Omega$ corresponds to a plane in $PG(3,\mathbb{F})$. 
This correspondence is known as the Klein correspondence. 

A spread is called Desarguesian if the translation plane it defines is Desarguesian, that is, if the Theorem of Desargues is valid in that plane \cite{Hirschfeld}. This fact justifies the name Desarguesian spread. A spread is Desarguesian if and only if it is isomorphic to a spread constructed with Segre's construction \cite{Segre}. For other characteristics of Desarguesian spreads, see for example \cite{Beutelspacher}. 
It is well-known that a Desarguesian spread of lines in $PG(3,\mathbb{F})$ corresponds in the Pl\"ucker embedding to the points of the complete intersection of the Klein quadric $Q$ and a projective space $U$ of three dimensions. For the purpose of this article, this construction could as well have been taken as the defining property of Desarguesian spreads.  
Below we provide a short detailed description of the correspondence between the lines in a spread of lines and the points in the intersection $U\cap Q$. 

For any line $\ell$ which does not intersect $Q$, consider the polar three-dimensional space $U=\ell^{\perp}$ of $\ell$ with respect to the quadric $Q$. (The points of $U$ are the points with vector representatives that are perpendicular to the vector representatives of all the points on $\ell$ under the bilinear form defining $Q$.)
Then the intersection of $U$ and $Q$ is a non-singular elliptic quadric in three dimensions, and any elliptic quadric contained in $Q$ can be obtained in this way. 
In particular, $U\cap Q$ does not contain lines. 
It is easy to see that this implies that the lines it represents in $PG(3,\mathbb{F})$ do not intersect and that they cover all the points in $PG(3,\mathbb{F})$. 
Indeed, two lines in $PG(3,\mathbb{F})$ that intersect in a point have, by the Klein correspondence, Pl\"ucker coordinates that are collinear on a line contained in $Q$. But the intersection $U\cap Q$ does not contain lines, so $U\cap Q$ does not contain Pl\"ucker coordinates of two lines with non-empty intersection. 
To see that each point of $PG(3,\mathbb{F})$ is on at least one of these lines, note that in the Klein correspondence each point in $PG(3,\mathbb{F})$ is represented by a plane contained in $Q$, and this plane must intersect the three-space $U$ in at least a point. Together this shows that the set of lines with Pl\"ucker coordinates in $U\cap Q$ is a spread of $PG(3,\mathbb{F})$. 

\begin{lemma}
\label{lemma:2}
If $S$ is a spread of lines in $PG(3,\mathbb{F})$ with primary Pl\"ucker coordinates in the complete intersection $U\cap Q$ of the Klein quadric $Q$ with a three-dimensional projective subspace $U$ of $PG(5,\mathbb{F})$, then the dual Pl\"ucker coordinates of $S$ are also in $U\cap Q$. 
\end{lemma}
\begin{proof}
This is a consequence of the relation between primary and dual coordinates. 
\end{proof}

The representation of a Desarguesian spread in the Klein quadric as $U\cap Q$ gives a very simple decoding algorithm.

Let $C\subseteq G_{\mathbb{F}}(2,4)$ be the spread code defined by the intersection of $Q$ with a suitable three-dimensional projective subspace $U$ of $PG(5,\mathbb{F})$. Let $$\left\{\begin{array}{ccc}a_0X_0+a_1X_1+a_2X_2+a_3X_3+a_4X_4+a_5X_5&=&0\\b_0X_0+b_1X_1+b_2X_2+b_3X_3+b_4X_4+b_5X_5&=&0\end{array}\right.$$ be the equations defining $U$.  
Assume that a line $c\in C$ is sent  and received as a subspace $x$ with at most one error. 
The subspace $x$ is represented in the form of a set of points of $PG(3,\mathbb{F})$ spanning a projective space of dimension either zero, one or two. 
In the first case, $x$ is a point contained in the line $c$. 
In the second case, $x$ equals the sent codeword $c$. 
Then no decoding is needed. 
In the third case, $x$ is a plane containing $c$. 

\begin{theorem}
\label{theorem:1}
Assume that the received subspace $x$ is a point $p=(p_0:p_1:p_2:p_3)$. 
Then the line $c$ that was sent is defined by the equations
$$\left\{\begin{array}{lcl} 
(-a_0p_1-a_1p_2-a_2p_3)X_0+(a_0p_0-a_3p_2-a_4p_3)X_1&&\\+(a_1p_0+a_3p_1-a_5p_3)X_2+(a_2p_0+a_4p_1+a_5p_2)X_3&=&0\\\\
(-b_0p_1-b_1p_2-b_2p_3)X_0+(b_0p_0-b_3p_2-b_4p_3)X_1&&\\+(b_1p_0+b_3p_1-b_5p_3)X_2+(b_2p_0+b_4p_1+b_5p_2)X_3&=&0\end{array}\right.$$
\end{theorem}

\begin{proof}
Calculate the Pl\"ucker coordinates for the line spanned by the points $p$ and $X=(X_0:X_1:X_2:X_3)$ and apply the equations of $U$. 
\end{proof}

\begin{theorem}
\label{theorem:2}
Assume that the received subspace $x$ is a plane defined by equations $p_0^*X_0+p_1^*X_1+p_2^*X_2+p_3^*X_3=0$. 
Then the line $c$ that was sent is spanned by the points
\bigskip\\
$(-a_0p_1^*-a_1p_2^*-a_2p_3^*:a_0p_0^*-a_3p_2^*-a_4p_3^*:a_1p_0^*+a_3p_1^*-a_5p_3^*:a_2p_0^*+a_4p_1^*+a_5p_2^*)$
\bigskip\\and\bigskip\\
$(-b_0p_1^*-b_1p_2^*-b_2p_3^*:b_0p_0^*-b_3p_2^*-b_4p_3^*:b_1p_0^*+b_3p_1^*-b_5p_3^*:b_2p_0^*+b_4p_1^*+b_5p_2^*).$
\end{theorem}
\begin{proof}
Dualize Theorem \ref{theorem:1} and use Lemma \ref{lemma:2}. 
\end{proof}
Theorem \ref{theorem:1} and Theorem \ref{theorem:2} give an algorithm for decoding spread codes in $G_{\mathbb{F}}(2,4)$. 
\begin{algorithm}
\caption{Let $S=U\cap Q$ be a line spread code defined by the complete intersection of a projective subspace $U$ of dimension three and the Klein quadric $Q$. 
Given a received subspace $x\in PG(3,\mathbb{F})$, which was sent as $c\in S$ and is given as a set of vectors spanning $x$, this algorithm calculates $c$. }
\label{alg:0}
\begin{algorithmic}
\State Use Gaussian elimination to see if $x$ has vector space dimension one, two or three. 
\State Projectively, $x$ is then a point, a line, or a plane.
\State To calculate the sent line $c$:
\If{$x$ is a point} 
\State use Theorem \ref{theorem:1} to find $c$; 
\EndIf
\If{$x$ is a line} 
\State $c:=x$;
\EndIf 
\If{$x$ is a plane} 
\State use Theorem \ref{theorem:2} to find $c$.
\EndIf
\State Return $c$;
\end{algorithmic}
\end{algorithm}

The complexity of this decoding algorithm is low. Indeed, when the received subspace has vector space dimension three, all that is needed is to calculate the coefficients of the sent subspace, requiring 24 multiplications and 16 additions.
The same number of operations is required to calculate the coefficients of the equations defining the sent subspace in case the sent subspace has vector space dimension one.  
If we assume that decoding a subspace means returning a basis of the sent subspace, in this case we are also required to solve that system of two linear equations, requiring at most 15 multiplications and 7 additions. 
Note that the equations of $U$ can be chosen carefully so that most of the $a_i$ and the $b_i$ are zero. This substantially reduces the number of operations further.  For example, by choosing $a_0=a_5=1$ and $b_1=-b_4=1$ and $a_i=b_j=0$ otherwise, the 24 multiplications in the calculation of the coefficients are reduced to 8 and no additions are needed. 

The complexity of determining the dimension of the received subspace is larger. 
In the case of network coding it is possible for a sink in a network to receive a large number of vectors, of which perhaps only a few are linearly independent. For decoding purposes it would typically be assumed that these vectors are already reduced to a set of linearly independent vectors. The complexity of this procedure is at most the complexity of Gaussian elimination of the system formed by the received vectors. The complexity of Gaussian elimination of a matrix of size $k\times n$, with $k\leq n$, is $k^2n$.  


\section{Decoding geometric subspace codes with explicit Schubert calculus}
\label{section:3}
In this section we present an error-correction decoding algorithm of geometric subspace codes of constant dimension $k$ in $PG(2k+1,\mathbb{F})$.  The algorithm relies on explicit Schubert calculus. 
Some preliminaries are needed.  
\subsection{The Pl\"ucker embedding of the Grassmannian}
An element $x$ of the Grassmannian $G_{\mathbb{F}}(k+1,n+1)$ is a projective subspace of dimension $k$ of $PG(n,\mathbb{F})$. It is spanned by $k+1$ points in general position corresponding to $k+1$ linearly independent vectors of $V(n+1,\mathbb{F})$. If $n+1=2(k+1)$, then dually, $x$ is the intersection of $n-k=k+1$ hyperplanes in general position, corresponding to $k+1$ linearly independent vectors of the dual space  $V^*(n+1,\mathbb{F})$. 
The exterior product of the vectors and dual vectors representing the $k+1$ points and the $k+1$ hyperplanes defines the primary and the dual Pl\"ucker coordinates of $x$, respectively, as the projectivization of vectors of the exterior algebra $\bigwedge^{k+1}V(n+1,\mathbb{F})$. 
If $\left(e_i\right)_i$ is a basis of $V(n+1,\mathbb{F})$, then $\left(\hat{e}_{\bf i}\right)_{\bf i}=\left(e_{i_0} \wedge \cdots \wedge e_{i_{k}}\right)_{(i_0,\dots,i_{k})}$  is a basis of $\bigwedge^{k+1}V(n+1,\mathbb{F})$, where ${\bf{i}}=(i_0,\dots,i_{k})$ goes through the combinations of distinct numbers between $0$ and $n$. We order the basis elements in lexicographic order. If the $k+1$ points spanning $x$ are $a_0,\dots,a_{k}$, consider the  matrix $A=(a_i)$ of dimension $(k+1)\times (n+1)$ with $a_0,\dots,a_{k}$ as row vectors. The primary coordinates of $m$ in the basis $\left(\hat{e}_{(i_0,\dots,i_{k})}\right)$ are then the projectivization of the determinants of the $\binom{n+1}{k+1}$ minors of dimension $(k+1)\times(k+1)$  of $A$, choosing columns in the same order as the basis. The dual coordinates are calculated analogously from the coefficients of equations defining the hyperplanes. For details see for example \cite{Hodge,Kleiman}.  
The primary and the dual coordinates are related so that, apart from some sign changes,  one set of coordinates is obtained from the other by reversing the order.  
\begin{lemma}
\label{lemma:2a}
\cite{Hodge}Let $x$ be a subspace of $G_{\mathbb{F}}(k+1,n+1)$. If $$(p_{{\bf i}_0}:\dots:p_{{\bf i}_N})$$ are the primary Pl\"ucker coordinates of $x$ , then the dual Pl\"ucker coordinates of $x$ are $$(p^{{\bf i'}_0}:\dots:p^{{\bf i'}_N})=(\epsilon_{{\bf i}_N}p_{{\bf i}_N}:\dots:\epsilon_{{\bf i}_0}p_{{\bf i}_0}),$$ where 
\begin{itemize}
\item $\epsilon_{\bf i}=1$ if $(i_0,\dots,i_k,i_{k+1},\dots,i_n)$ is an even permutation, and 
\item  $\epsilon_{\bf i}=-1$ if $(i_0,\dots,i_k,i_{k+1},\dots,i_n)$ is an odd permutation, 
\end{itemize}
where ${\bf i}=(i_0,\dots,i_k)$ and ${\bf i'}=(i_{k+1},\dots,i_n)$  are in lexicographic order and $N=\binom{n+1}{k+1}-1$. 
\end{lemma}

It is well-known that this implies that the image of $G_{\mathbb{F}}(k+1,n+1)$ under the Pl\"ucker embedding is an algebraic variety $\mathcal{V}$ of $PG\left(\binom{n+1}{k+1}-1,\mathbb{F}\right)$ defined by the intersection of quadrics. The quadratic equations defining these quadrics are sometimes called the Pl\"ucker relations. For more details, see for example \cite{Hodge,Kleiman}. 

\subsection{Schubert varieties}
A \emph{flag} is a sequence of nested projective subspaces $F:A_0\subsetneq A_1\subsetneq \cdots\subsetneq A_k$ of $PG(n,\mathbb{F})$. The \emph{Schubert variety} defined by $F$ is the set $\Omega(F)$ of projective subspaces $X$ of dimension $k$ satisfying $\dim(X\cap A_i)\geq i$. 
Therefore, a Schubert variety is a set of points in the Pl\"ucker embedding $\mathcal{V}\subseteq PG\left(\binom{n+1}{k+1}-1,\mathbb{F}\right)$ of $G_{\mathbb{F}}(k+1,n+1)$. 
For example, the received point $p$ in Section \ref{section:2} defines a flag $F:p\subsetneq PG(3,\mathbb{F})$. The Schubert variety of this flag is the set of lines passing through $p$, with Pl\"ucker coordinates forming a plane contained in the Klein quadric. 

It is well-known that a Schubert variety is the intersection of a projective subspace $W$ and $\mathcal{V}$ \cite{Hodge,Kleiman}. 
The two classes of planes contained in the Klein quadric, corresponding to the points and the planes of $PG(3,\mathbb{F})$ under the Klein correspondence, give an interesting example of when $W$ is contained in $\mathcal{V}$. 
In general the subspace $W$ is not contained in $\mathcal{V}$.  
The Klein correspondence generalizes to flags $F:A_0\subsetneq A_1\subsetneq \cdots\subsetneq A_k$ in $PG(2k+1,\mathbb{F})$ where $A_{k-1}$ has dimension $k-1$, $A_k=PG(2k+1,\mathbb{F})$ and, for the rest of the indices, $A_i$ is any subspace of dimension $i$ contained in $A_{i+1}$. This is described in Lemma \ref{lemma:3}, and can be found for example in \cite{Hodge}. 

\begin{lemma}
\label{lemma:3}
\begin{enumerate}
\item The set of $k$-dimensional projective subspaces of $PG(2k+1,\mathbb{F})$ intersecting in a fixed $(k-1)$-dimensional projective subspace, corresponds to the points of a $(k+1)$-dimensional projective subspace contained in the Pl\"ucker embedding $\mathcal{V}$ of the Grassmannian $G_{\mathbb{F}}(k+1,2k+2)$. 

\item Dually, the set of $k$-dimensional projective subspaces of $PG(2k+1,\mathbb{F})$ which are contained in a fixed $(k+1)$-dimensional projective subspace, correspond to the points of a $(k+1)$-dimensional projective subspace contained in the Pl\"ucker embedding $\mathcal{V}$ of the Grassmannian $G_{\mathbb{F}}(k+1,2k+2)$. 
\end{enumerate}
\end{lemma}
\begin{proof}
\begin{enumerate}
\item Let $x=\left\langle x_0,\cdots, x_{k-1}\right\rangle$ be a vector subspace representing a $(k-1)$-dimensional projective subspace of $PG(2k+1,\mathbb{F})$. Then the $k$-dimensional projective subspaces containing $x$ have Pl\"ucker coordinates of the form $x\wedge v=(x_0\wedge\cdots\wedge x_{k-1})\wedge v$ where $v$ is the vector representative of a point of $PG(2k+1,\mathbb{F})$. Fix a basis $V(2k+2,\mathbb{F})=\left\langle e_0=x_0,\dots,e_{k-1}=x_{k-1}, e_{k},\dots, e_{2k+1}\right\rangle$ and write $v=a_{k}e_{k}+\cdots+a_{2k+1}e_{2k+1}$. 
Then $x\wedge v= x\wedge (a_{k}e_{k}+\cdots+a_{2k+1}e_{2k+1}) = a_{k}x\wedge e_{k}+\cdots+a_{2k+1}x\wedge e_{2k+1}$, which is a linear combination of $k+2$ linearly independent vectors, the representatives of $k+2$ points of $\mathcal{V}$ spanning a projective subspace of dimension $k+1$ in $PG\left(\binom{2k+2}{k+1}-1,\mathbb{F}\right)$. 

If $y$ is a vector in the subspace $\left\langle x\wedge e_{k},\cdots, x\wedge e_{2k+1}\right\rangle$, then $y=c_kx\wedge e_{k}+\cdots + c_{2k+1}x\wedge e_{2k+1}=x\wedge (c_ke_k+\cdots + c_{2k+1}e_{2k+1}$, giving a point in $\mathcal{V}$, and so the entire projective subspace of dimension $k+1$ is contained in $\mathcal{V}$.
\item Dualize and use Lemma \ref{lemma:2a}.  
\end{enumerate} 
\end{proof}

More generally, the set of $k$-dimensional projective subspaces of $PG(2k+1,\mathbb{F})$ which intersect in a fixed $(k-1)$-dimensional projective subspace and are contained in a fixed $(k+s)$-dimensional projective subspace, correspond to the points of an $s$-dimensional projective subspace contained in $\mathcal{V}$. The proof of Lemma \ref{lemma:3} can be used to prove this, by choosing $v$ so that it is contained in the given $(k+s)$-dimensional projective subspace. The dual statement is also true.

\subsection{Explicit Schubert varieties for decoding}
\label{sec:explicit}
Let $\mathcal{C}$ in $PG(n,\mathbb{F})$ be a subspace code correcting $t$ errors, and assume that a subspace $c\in \mathcal{C}$ is sent. 
Given a received subspace $x$ of subspace distance at most $t$ from $c$, the decoding problem is to calculate $c$. 

An algebraic variety is a set of points such that their coordinates satisfy a set of polynomial equations. Consider a subspace code $\mathcal{C}$ of constant projective dimension $k$ in $PG(n,\mathbb{F})$, defined geometrically as an algebraic variety contained in the Pl\"ucker embedding of the Grassmannian $G_{\mathbb{F}}(k+1,n+1)$. We call such a code a \emph{geometric subspace code} of projective dimension $k$ in $PG(n,\mathbb{F})$. 
\begin{lemma}
Let $\mathcal{C}$ be a geometric subspace code of projective dimension $k$ in $PG(n,\mathbb{F})$. 
The set of codewords $\{c\}$ in $\mathcal{C}$ such that the intersection $c\cap x$ has smallest expected dimension $\delta$ can be calculated as the intersection of the code variety $C$ and the Schubert variety $\Omega(F(x,\delta,k))$ defined by any flag $F(x,\delta,k): A_0\subsetneq \cdots \subsetneq A_{k}$ such that $A_{\delta}=x$, $A_k=PG(n,\mathbb{F})$, and $\dim(A_i)=\dim(A_{i+1})-1$ for the the rest of the $i$. 
\end{lemma}
\begin{proof}
Any flag with these characteristics will define the same Schubert variety, namely the Schubert variety in which the points represent the subspaces of $G_{\mathbb{F}}(k+1,n+1)$ intersecting $x$ in a subspace of dimension at least $\delta$. Indeed, if $y$ is a projective subspace of dimension $k$, then $y$ intersects any subspace $A_i$ of projective dimension $n-(k-i)$ in a subspace of dimension at least $i$. Note that, for $\delta<i\leq k$, the subspaces in $F(x,\delta,k)$ are defined so that $\dim(A_i)=n-(k-i)$. 
Of these subspaces $y$, the ones intersecting $x$ in a subspace of projective dimension less than $\delta$ are excluded by choosing $A_{\delta}=x$. Any subspace $y$ intersecting $x$ in a subspace of projective dimension at least $\delta$ will intersect any subspace $A_i\subsetneq x$ of dimension $\dim(A_i)=\dim(x)-(\delta-i)$ in a subspace of at least dimension $i$. Note that, for $0\leq i\leq \delta$, the subspaces in $F(x,\delta,k)$ are defined such that $\dim(A_i)=\dim(x)-(\delta-i)$.  
\end{proof}
Next we express $\delta$ in $F(x,\delta,k)$ in terms of the largest expected subspace distance between the received subspace $x$ and the sent subspace $c$, that is, in terms of the error-correcting capacity $t$ of the code.  
\begin{lemma}
If $c$ is a projective subspace of dimension $k$ and $x$ is a projective subspace at subspace distance from $c$ satisfying $d(x,c)\leq t$, then $\dim(c\cap x)\geq ((k-t)+\dim(x))/2$. In particular, if $k=t$, then $\dim(c\cap x)\geq \dim(x)/2$. 
\end{lemma}
\begin{proof}
The subspace distance between $c$ and $x$ is $d(x,c)=\dim(c)+\dim(x)-2\dim(c\cap x)$. Therefore $d(x,c)\leq t$ implies $\dim(c)+\dim(x)-2\dim(c\cap x)\leq t$, so that $\dim(c \cap x)\geq (\dim(c)+\dim(x)-t)/2=(k-t+\dim(x))/2$. If $k=t$, then we get $\dim(c\cap x)\geq \dim(x)/2$. 
\end{proof}
The Schubert variety to use for decoding is therefore in the general case $\Omega(F(x,(k-t+\dim(x))/2,k))$ and for subspace codes in $G_{\mathbb{F}}(t+1,2t+2)$ correcting $t$ errors, like $t$-spread codes, it is  $\Omega(F(x, \dim(x)/2,t))$.  

Decoding requires a method for calculating $\Omega(F(x,\delta,k))$ explicitly. 
As mentioned before, a Schubert variety is the intersection of a projective subspace $\mathcal{W}$ and the image of $G_{\mathbb{F}}(k+1,n+1)$ under the Pl\"ucker embedding. 
The following is a description of $\mathcal{W}(F(x,\delta,k))$ which can be used in the design of the decoding algorithm.


\begin{lemma}
\label{lemma:schubertvariety}

Let 
\begin{itemize}
\item $\{e_0,\dots,e_n\}$ be a basis of $V=V(n+1,\mathbb{F})$,  
\item $x=\left\langle x_0,\dots,x_{\dim(x)}\right\rangle \subseteq V(n+1,\mathbb{F})$ be a vector subspace of dimension $\dim(x)\leq k$ when regarded as a projective space, and
\item $\delta \leq \dim(x)$. 
\end{itemize}
Then $\Omega(F(x,\delta,k)) \subseteq  PG\left(\binom{n+1}{k+1}-1, \mathbb{F} \right)$ can be calculated as $$\Omega(F(x,\delta,k))=\mathcal{W}(F(x,\delta,k))\cap \mathcal{V},$$ where 
\begin{itemize}
\item $\mathcal{V}$ is the Pl\"ucker embedding of the Grassmannian $G_{\mathbb{F}}(k+1,n+1)$ and 
\item $\mathcal{W}(F(x,\delta,k))$ is the projectivization of the vector space $W(F(x,\delta,k))=\left\langle w_{({\bf i},{\bf j})} \right\rangle_{({\bf i},{\bf j})}$, where $$w_{({\bf i},{\bf j})}=x_{i_{0}}\wedge \cdots \wedge x_{i_{\delta}}\wedge e_{j_{\delta+1}}\wedge \cdots \wedge e_{j_k},$$ and 
\begin{itemize}
\item ${\bf i}=(i_{0},\dots,i_{\delta})$ goes through all $(\delta+1)$-combinations of $\{0,\dots,\dim(x)\}$, and 
\item ${\bf j}=(j_{\delta+1},\dots,j_{k})$ goes through all $(k-\delta)$-combinations of $\{0,\dots,n\}$. 
\end{itemize}
\end{itemize}
\end{lemma}

\begin{proof}
A vector $v$ of $\bigwedge^{k+1}V$ is called \emph{totally decomposable} if it can be written as $v=v_0\wedge\cdots\wedge v_k$ for some vectors $v_0,\dots,v_k\in V$. The set of totally decomposable vectors of $\bigwedge^{k+1}V$ are exactly the points in the Pl\"ucker embedding $\mathcal{V}$ of the Grassmannian. 

Consider the restriction of the wedge product of $\bigwedge^{\delta+1}V$ and $\bigwedge^{k-\delta}V$ to $\bigwedge^{\delta+1}x$ and $\bigwedge^{k-\delta}V$. This wedge product is a bilinear map $f:\bigwedge^{\delta+1}x\times \bigwedge^{k-\delta}V \rightarrow \bigwedge^{k+1}V$ defined in terms of the basis vectors as $$f\left(x_{i_0}\wedge \cdots \wedge x_{i_{\delta}}, e_{j_{\delta +1}}\wedge \cdots \wedge e_{j_{k}}\right)=x_{i_0}\wedge \cdots \wedge x_{i_{\delta}} \wedge e_{j_{\delta+1}}\wedge \cdots \wedge e_{j_{k}}$$ for $(i_0,\dots,i_{\delta})$ in the $(\delta+1)$-combinations of $\{0,\dots,\dim(x)\}$ and $(j_{\delta+1},\dots,j_{k})$ in the $(k-\delta)$-combinations of $\{0,\dots,n\}$. 
Let $W$ be the subspace of $\bigwedge^{k+1}V$ spanned by the image of $f$, that is, $W=\left\langle f\left(\bigwedge^{\delta+1}x\times \bigwedge^{k-\delta}V\right)\right\rangle$. 
Denote by $W'$ the subspace of $\bigwedge^{k+1}V$ which is spanned by the image of the Cartesian product of the basis vectors of $\bigwedge^{\delta+1}x$ and $\bigwedge^{k-\delta}V$, namely $$W'=\left\langle f\left(\{ x_{i_0}\wedge \cdots\wedge x_{i_{\delta}}\}_{{\bf i}}\times \{e_{j_{\delta+1}}\wedge \cdots \wedge e_{j_{k}}\}_{{\bf j}}\right)\right\rangle.$$ 
We want to show that $W=W'$. It is clear that $W'\subseteq W$. 
To see that $W\subseteq W'$ it is enough to see that all elements in the image of $f$ belongs to $W'$. 
If $u$ belongs to the image of $f$, then $u=f\left(\sum_{\bf i}a_{\bf i}x_{i_0}\wedge \cdots \wedge x_{i_{\delta}}, \sum_{\bf j} b_{\bf j}e_{j_{\delta+1}}\wedge \cdots \wedge e_{j_{k}}\right)=\sum_{({\bf i},{\bf j})}a_{\bf i}b_{\bf j}x_{i_0}\wedge \cdots \wedge x_{i_{\delta}} \wedge e_{j_{\delta+1}}\wedge \cdots \wedge e_{j_{k}}$ for some $\sum_{\bf i}a_{\bf i}x_{i_0}\wedge \cdots \wedge x_{i_{\delta}}\in  \bigwedge^{\delta+1}x$ and some $\sum_{\bf j} b_{\bf j}e_{j_{\delta+1}}\wedge \cdots \wedge e_{j_{k}}\in \bigwedge^{k-\delta}V$, so $u$ is a linear combination of the elements in the set $f\left(\{ x_{i_0}\wedge \cdots\wedge x_{i_{\delta}}\}_{{\bf i}}\times \{e_{j_{\delta+1}}\wedge \cdots \wedge e_{j_{k}}\}_{{\bf j}}\right)$. So $u\in W'$, implying that $W\subseteq W'$. 

Finally, we observe that $v$ is a totally decomposable vector of $W$ if and only if $v=v_0\wedge \cdots \wedge v_k$, such that $\langle v_0,\dots,v_k\rangle \subseteq G_{\mathbb{F}}(k+1,n+1)$ and the vector space dimension satisfies $\dim(\langle v_0,\dots,v_k\rangle\cap x)\geq \delta+1$ (the projective dimension of the intersection is larger than $\delta$). 


To conclude, if $\mathcal{W}$ is the projectivization of $W$, then $\mathcal{V}\cap \mathcal{W}$ is the set of Pl\"ucker coordinates of the subspaces of $G_{\mathbb{F}}(k+1,n+1)$ intersecting $x$ in a subspace of projective dimension at least $\delta$. This is exactly the Schubert variety $\Omega(F(x,\delta,k))$. 
\end{proof}

\begin{lemma}
\label{lemma:dim}
The vector space dimension of $W(F(x,\delta,k))$ is $$\sum_{d=\delta}^{\dim(x)}\binom{\dim(x)+1}{d+1}\binom{n-\dim(x)}{k-d}.$$
\end{lemma}
\begin{proof}
We may assume that $x=\langle e_0,\dots,e_{\dim(x)}\rangle$, where $\dim(x)$ is the projective dimension of $x$. For each vector space dimension $d+1\in \{\delta+1,\dots,\dim(x)+1\}$, the number of subspaces of $V(n+1,\mathbb{F})$ spanned by a subset of the basis vectors and intersecting $x$ in a subspace of dimension exactly $d+1$, is the number of subspaces of dimension $d+1$ of $x$ spanned by a subset of the basis vectors, times the number of subspaces of dimension $k-d$  spanned by a subset of the basis vectors of the orthogonal complement of $x$, that is, the subspace $x^{\perp}$ such that $x\oplus x^{\perp}=V(n+1,\mathbb{F})$), hence $\binom{\dim(x)+1}{d+1}\binom{n-\dim(x)}{k-d}$. 
\end{proof}

Below several methods for calculating the subspace $\mathcal{W}(F(x,\delta,k))$ defining $\Omega(F(x,\delta,k))$ are described.  
%
%
The most efficient of these methods are based on the idea to first describe $\mathcal{W}(F(E,\delta,k))$ for a particular  subspace $E$, and then use a base change to find a description of $\mathcal{W}(F(x,\delta,k))$. 

In $V(n+1,\mathbb{F})$, consider the subspace $E_{b}$ spanned by the $b+1$ first basis vectors. 
By Lemma \ref{lemma:schubertvariety} and \ref{lemma:dim}, the Schubert variety representing the set of subspaces of projective dimension $k$ intersecting $E_b$ in a subspace of projective dimension at least $\delta$ is spanned by the Pl\"ucker coordinates of the subspaces of dimension $k$ which are obtained by adding $k-i$ vectors from the orthogonal complement of $E_b$ to a subset of $i+1$ of the basis vectors of $E_b$ for $i\in \{\delta,\dots, b\}$. 

The space $G_{\mathbb{F}}(k+1,n+1)$ is a homogeneous space under the transitive action of the general linear group of invertible $(n+1)\times (n+1)$ matrices, $GL(n+1,\mathbb{F})$, representing the base changes in $V(n+1,\mathbb{F})$. 
Therefore any $x\in G_{\mathbb{F}}(b+1,n+1)$ can be obtained from $E_b$ through a change of basis defined by a matrix in $GL(n+1,\mathbb{F})$. 
As there are many bases for $x$ and $E_b$,  there are many changes of basis transforming $E_b$ to $x$. If we are given $x=\left\langle x_0,\dots,x_b\right\rangle$, then we may, for example, choose $B$ to be a change of basis transforming the first $b+1$ vectors $e_0,\dots,e_{b}$ of the standard basis of $V(n+1,\mathbb{F})$ to $x_0, \dots, x_b$. It is described in \cite{Rosenthal} how $B$ can be chosen to make calculations efficient: calculate the reduced row-echelon form of the matrix with row vectors a basis of $x$. Add $n-b$ new rows to this matrix by taking distinct vectors from the standard basis of $V(n+1,\mathbb{F})$ such that the columns under each pivot element are still zero. Now, if $M_{E_b}$ and $M_x$ are the matrices with row vectors the basis of $E_b$ and $x$ respectively, then $M_{E_b}B=M_{x}$, that is, $B$ gives the change of basis from a basis in which the $b+1$ first base vectors is a basis of $x$ (so that $x$ in that basis is $E_b$) to another basis in which the $b+1$ first base vectors are from the standard basis. 

A change of basis in $V(n+1,\mathbb{F})$ induces a change of basis in the exterior algebra $\bigwedge^{k+1}V(n+1,\mathbb{F})$. 
The new basis of $\bigwedge^{k+1}V(n+1,\mathbb{F})$ is the wedge product of the $(k+1)$-combinations of the new basis vectors of $V(n+1,\mathbb{F})$. 
Each wedge product implies the calculation of $\binom{n+1}{k+1}$ minors of dimension $(k+1)\times (k+1)$.   
The computational complexity for calculating the matrix for the change of basis in $\bigwedge^{k+1}V(n+1,\mathbb{F})$, given the matrix for the change of basis in $V(n+1,\mathbb{F})$, is therefore in the worst case of order $O\left(\binom{n+1}{k+1}^2(k+1)^3\right)$, since the complexity for calculating one minor of dimension $(k+1)\times(k+1)$ is $(k+1)^3$. 

\subsubsection{Parametrization}
\label{sec:param}
First we describe how to give a parametrization of $\mathcal{W}(F(x,\delta,k))$, where the projective dimension of $x$ is $b$, from a parametrization of $\mathcal{W}(F(E_b,\delta,k))$, where $E_b$ is the subspace spanned by the $b+1$ first basis vectors of $V(n+1,\mathbb{F})$. 
According to Lemma \ref{lemma:schubertvariety} and Lemma \ref{lemma:dim}, a  basis of $W(F(E_b,\delta,k))$ can be obtained by calculating $e_{i_0}\wedge \cdots \wedge e_{i_{d}}\wedge e_{j_{d+1}}\wedge \cdots \wedge e_{j_{k}}$ for ${\bf i}=(i_{0},\dots,i_{d})$ goes through all $(d+1)$-combinations of $\{0,\dots,b\}$, and ${\bf j}=(j_{d+1},\dots,j_{k})$ goes through all $(k-d)$-combinations of $\{b+1,\dots,n\}$, for $d\in \{\delta,\dots,b\}$. 
But the vectors in this basis of $W(F(E_b,\delta,k))$ are exactly the vectors in the standard basis of $\bigwedge^{k+1}V(n+1,\mathbb{F})$ with indices $({\bf i, j})=(i_{0},\dots,i_{d}, j_{d+1},\dots,j_{k})$, so no calculations are needed. A parametrization of $\mathcal{W}(F(E_b,\delta,k))$ is therefore $$p(\alpha)=\mathbb{P}\left(\sum_{\bf (i,j)} \alpha_{\bf (i,j)} e_{i_0}\wedge \cdots \wedge e_{i_{d}}\wedge e_{j_{d+1}}\wedge \cdots \wedge e_{j_{k}}\right),$$
where the parameters are $\alpha=(\alpha_{\bf (i,j)})$.  
To obtain a parametrization of $\mathcal{W}(F(x,\delta,k))$ from the parametrization of $\mathcal{W}(F(E_b,\delta,k))$, use the matrix $B$ for change of basis from the standard basis to a basis where the first $b+1$ vectors $e_0,\dots,e_{b}$ of the standard basis of $V(n+1,\mathbb{F})$ are transformed to $x_0, \dots, x_b$, where $\{x_0,\dots,x_b\}$ is  a basis of $x$. 
A basis of $W(F(x,\delta,k))$ is then simply obtained by choosing the columns of this matrix with indices $({\bf i, j})=(i_{0},\dots,i_{\delta}, j_{\delta+1},\dots,j_{k})$. 
No matrix multiplication is needed. 
Note that these multi-indices ${\bf (i,j)}=(i_{0},\dots,i_{d},j_{d+1},\dots,j_{k})$ are exactly the ones satisfying $(i_{0},\dots,i_{d},j_{d+1},\dots,j_{k}) \preceq (b-\delta, \dots, b, n-b, \dots, n)$, where $\preceq$ is the partial order defined so that $(a_1,\dots,a_m)\preceq (b_1,\dots,b_m)$ whenever $a_s\leq b_s$ for all $s\in\{1,\dots,m\}$.

\begin{algorithm}[H]
\caption{Given a subspace $x=\left\langle x_0,\dots,x_{\dim(x)}\right\rangle \subseteq V(n+1,\mathbb{F})$ of vector dimension $\dim(x)+1\leq k+1$ and a projective dimension $\delta \leq \dim(x) \leq k$, this algorithm calculates a basis of $W(F(x,\delta,k)) \subseteq  \bigwedge^{k+1}V(n+1,\mathbb{F})$.}
\label{alg:coord3}
\begin{algorithmic} 
\State Let $\{v_0,\dots,v_n\}$ be the row vectors of the matrix $B$ calculated as in \cite{Rosenthal} (as also described briefly above).
\For{the multi-indices ${\bf i}=(i_0,\dots,i_k)$ with $0<i_j<i_{j+1}<n$ such that ${\bf i} \preceq  (b-\delta, \dots, b, n-b, \dots, n)$} 
\State Calculate the vectors $v_{\bf i}=v_{i_0}\wedge \cdots \wedge v_{i_k}$. 
\EndFor
\State Return $\{v_{\bf i}\}$. 
\end{algorithmic}
\end{algorithm}
A parametrization of $\mathcal{W}(F(x,\delta,k))$ is now $p(\alpha)=\mathbb{P}\left(\sum_{\bf i} \alpha_{\bf i}v_{\bf i}\right), $
where the parameters are $\alpha=(\alpha_{\bf i})$.  

\begin{lemma}
\label{lem:comp1}
The complexity of Algorithm \ref{alg:coord3} is $$O\left(\binom{n+1}{k+1}(\delta+1)^{3}D\right),$$ where $D=\sum_{d=\delta}^{\dim(x)}\binom{\dim(x)+1}{d+1}\binom{n-\dim(x)}{k-d}$ is the dimension of $W$.
\end{lemma}
\begin{proof}
The algorithm implies the calculation of $$\dim(W)=\sum_{d=\delta}^{\dim(x)}\binom{\dim(x)+1}{d+1}\binom{n-\dim(x)}{k-d}$$ vectors in $\bigwedge^{k+1}V(n+1,\mathbb{F})$. 
Each vector has $\binom{n+1}{k+1}$ components, and each component is calculated from $\delta+1$ of the $b+1$ first rows and $n-\delta$ of the $n-b$ last rows of the matrix $B$. The $n-b$ last rows are vectors from the standard basis of $V(n+1,\mathbb{F})$. 
Therefore, by Laplace expansion along the $n-b$ last rows of the matrix with rows $\{v_0,\dots,v_n\}$, each of these components requires the calculation of a determinant of a matrix of dimension at most $(\delta+1)\times(\delta+1)$. 
The complexity of calculating the determinant of a $(\delta+1)\times(\delta+1)$-matrix is $O\left((\delta+1)^{3}\right)$. 
Therefore the overall complexity is $$O\left(\binom{n+1}{k+1}(\delta+1)^{3}\sum_{d=\delta}^{\dim(x)}\binom{\dim(x)+1}{d+1}\binom{n-\dim(x)}{k-d}\right).$$ 
\end{proof}

It is also possible to calculate a parametrization of $\mathcal{W}(F(x,\delta,k))$ by a direct application of Lemma \ref{lemma:schubertvariety}. 
Because in this case it cannot be assumed that a basis of $x^{\perp}$ is known, the number of Pl\"ucker coordinates to calculate can only be bounded by $\binom{\dim(x)+1}{\delta+1}\binom{n}{k-\delta}$.
Therefore a deterministic implementation of the resulting algorithm will in general  have higher complexity than Algorithm \ref{alg:coord3}.  
However the complexity can be reduced to a complexity comparable to Algorithm \ref{alg:coord3} by making the algorithm probabilistic. 
Given a set of linearly independent vectors in a vector space of dimension $n$, the probability that a randomly chosen vector is not contained in the span of these vectors is large as long as the dimension of the span is smaller than $n$ and the distribution is well-chosen. So with high probability, only $D=\sum_{d=\delta}^{\dim(x)}\binom{\dim(x)+1}{d+1}\binom{n-\dim(x)}{k-d}$ Pl\"ucker coordinates have to be calculated. However, since the complexity is of the same order as the complexity of Algorithm \ref{alg:coord3}, and the latter is deterministic, this probabilistic algorithm is less interesting and is not described in detail here. 

\subsubsection{Equations}
We now describe how to find the linear equations defining $\mathcal{W}(F(x,\delta,k))$ as a linear projective variety in $PG\left(\binom{n+1}{k+1}-1,\mathbb{F}\right)$. 
Two different methods are described, both intimately related to the method described in Section \ref{sec:param}.

The first method is essentially the same as the one described in \cite{Rosenthal2}, but it is included here for the sake of completeness, together with a calculation of its complexity. It is the dual of the method for finding a parametrization described in Section \ref{sec:param}. 



Consider the subspace $E_b$ as in Section \ref{sec:param}. 
As there, the Schubert variety representing the set of subspaces of projective dimension $k$ intersecting $E_b$ in a subspace of projective dimension at least $\delta$ is spanned by the Pl\"ucker coordinates of the subspaces of dimension $k$ which are obtained by adding $k-d$ vectors from the orthogonal complement of $E_b$ to a subset of $d+1$ of the basis vectors of $E_b$ for $d\in \{\delta,\dots, b\}$.  
The Pl\"ucker coordinates $(p_{{\bf i}_0}:\cdots:p_{{\bf i}_N})$ of these subspaces can only be (possibly) non-zero at the positions indexed by ${\bf i}=(i_0,\dots,i_n) \preceq (b-\delta, \dots, b, n-b, \dots, n)$, where, again, $\preceq$ is the partial order defined so that $(a_1,\dots,a_m)\preceq (b_1,\dots,b_m)$ whenever $a_s\leq b_s$ for all $s\in\{1,\dots,m\}$.  
At all other positions these Pl\"ucker coordinates are zero. 
This implies that $\mathcal{W}(F(E_b,\delta,k))$, defining the Schubert variety containing the points representing subspaces of vector dimension $k+1$ containing $E_a$ is defined by the equations 
$X_{\bf i}=0$, such that ${\bf i} \not \preceq  (b-\delta, \dots, b, n-k+\delta+1, \dots, n)$. 

\begin{lemma}
The number of equations defining $\mathcal{W}(F(E_b,\delta,k))$ is $$\binom{n+1}{k+1}-\sum_{d=\delta}^{b}\binom{b+1}{d+1}\binom{n-b}{k-d}.$$  
\end{lemma}
\begin{proof}
This number is $\binom{n+1}{k+1}-D$ where $D$ is the vector space dimension of $W(F(E_b,\delta,k))$, calculated in Lemma \ref{lemma:dim}. 
\end{proof}

Consider the vectors with components the coefficients of the equations $X_{\bf i}=0$ defining $\mathcal{W}(F(E_b,\delta,k))$.
Then each of these vectors is a vector of the standard basis for the dual space of the exterior algebra, $\left(\bigwedge^{k+1}V(n+1,\mathbb{F})\right)^*$. 
The coefficients of the equations defining $\mathcal{W}(F(x,\delta,k))$ can be obtained by a base change of $\left(\bigwedge^{k+1}V(n+1,\mathbb{F})\right)^*$ induced by a base change of $V(n+1,\mathbb{F})$ from a basis in which $x$ is written as $E_b$ to the standard basis.  
Again, no matrix multiplication is needed, only the calculation of $\binom{n+1}{k+1}-\sum_{d=\delta}^{b}\binom{b+1}{d+1}\binom{n-b}{k-d}$ of the basis vectors of the new basis of $\left(\bigwedge^{k+1}V(n+1,\mathbb{F})\right)^*$. 
\begin{algorithm}[H]
\caption{Given a subspace $x=\left\langle x_0,\dots,x_{\dim(x)}\right\rangle \subseteq V(n+1,\mathbb{F})$ of vector dimension $\dim(x)+1\leq k+1$ and a projective dimension $\delta \leq \dim(x) \leq k$, this algorithm returns the equations defining $\mathcal{W}(F(x,\delta,k))$.}
\label{alg:impl2}
\begin{algorithmic} 
\State Let $\{v_0,\dots,v_n\}$ be the row vectors of the matrix $B$ calculated as in \cite{Rosenthal} (as also described briefly above).
\For{the multi-indices ${\bf i}=(i_0,\dots,i_k)$ with $0<i_j<i_{j+1}<n$ such that ${\bf i} \not \preceq  (b-\delta, \dots, b, n-b, \dots, n)$} 
\State Calculate the vectors $v_{\bf i}=v_{i_0}\wedge \cdots \wedge v_{i_k}$. 
\EndFor
\State Return $\{v_{\bf i}\}$ as the coefficient vectors of the linear equations $v_{\bf i}X^t=0$, where $X=(X_0,\dots,X_{\binom{n+1}{k+1}-1})$. 
\end{algorithmic}
\end{algorithm}

\begin{lemma}
The computational complexity for calculating these vectors is $$O\left(\binom{n+1}{k+1}(\delta+1)^{3}\left(\binom{n+1}{k+1}-D\right)\right),$$ where $D=\sum_{d=\delta}^{b}\binom{b+1}{d+1}\binom{n-b}{k-d}$ is the dimension of $W(F(x,\delta,k))$.
\end{lemma}
\begin{proof}
Analogous to the proof of Lemma \ref{lem:comp1}, with the difference that this algorithm requires the calculation of $\binom{n+1}{k+1}-\dim(W)=\binom{n+1}{k+1}-D$ vectors in $\bigwedge^{k+1}V(n+1,\mathbb{F})$.  
\end{proof}

Given a parametrization of an algebraic variety, it is always possible to find equations defining the variety. This process is called \emph{implicitization}. 
The variety we are considering, $\mathcal{W}(F(x,\delta,k))$, is a smooth linear variety. Therefore, given equations defining it, it is always possible to find a (global) parametrization. 
This gives methods for finding equations and a parametrization of $\mathcal{W}(F(x,\delta,k))$, defined in terms of methods for finding a parametrization and equations, respectively. 
In particular, we obtain this second method for finding the equations defining $\mathcal{W}(F(x,\delta,k))$. 
 
\begin{algorithm}[H]
\caption{Given a subspace $x=\left\langle x_0,\dots,x_{\dim(x)}\right\rangle \subseteq V(n+1,\mathbb{F})$ of vector dimension $\dim(x)+1\leq k+1$ and a projective dimension $\delta \leq \dim(x) \leq k$, this algorithm returns the equations defining $\mathcal{W}(F(x,\delta,k))$.}
\label{alg:impl}
\begin{algorithmic} 
\State Use Algorithm \ref{alg:coord3} to calculate a parametrization of  $\mathcal{W}(F(x,\delta,k))$. 
\State Calculate and return the implicitation of this parametrization.  
\end{algorithmic}
\end{algorithm}

\begin{lemma}
The order of complexity of Algorithm \ref{alg:impl} is dominated by the order of complexity of Algorithm \ref{alg:coord3}.
\end{lemma}
\begin{proof}
The complexity of implicitation of linear equations is of polynomial order. 
\end{proof}

Similarly there is an algorithm which finds a parametrization of $\mathcal{W}(F(x,\delta,k))$ from the linear equations defining it, by solving the linear system of equations.

\subsection{A decoder of geometric subspace codes} 
\label{sec:decoder}
In Section \ref{sec:explicit} we announced the existence of an algorithm decoding geometric subspace codes by intersecting a Schubert variety  with the code variety. 
This algorithm was first described in \cite{Rosenthal, Rosenthal2}. 
Here new versions of this algorithm are presented, distinct from each other and from previous versions. 

The first version of the geometric decoding algorithm uses the parametrization of $\mathcal{W}(F(x,\delta,k))$ and the equations defining $\mathcal{C}$. 
Although not strictly necessary, this version of the algorithm is described here so that it requires the code to be a geometric subspace code in $G_{\mathbb{F}}(k+1,2k+2)$. This makes the algorithm dualize nicely. The restriction can be removed without much trouble. It is not present in Algorithm \ref{alg:4b} and Algorithm \ref{alg:4c}. 



\begin{algorithm}[H]
\caption{Let $\mathcal{C}$ be a geometric subspace code of constant (projective) dimension $k$ in $PG(2k+1,\mathbb{F})$, defined as $\mathcal{C}=\mathcal{V}\cap \mathcal{U}$, where $\mathcal{U}\subseteq PG\left(\binom{2k+2}{k+1}-1,\mathbb{F}\right)$ is the algebraic variety defining $\mathcal{C}$ and $\mathcal{V}$ is the Pl\"ucker embedding of the Grassmannian. 
Assume that the Pl\"ucker coordinates of $\mathcal{C}$ are defined by a set of equations $E_{\mathcal{C}}(X)$, where $X=(X_{\bf i})$ is a vector of variables of dimension $\binom{2k+2}{k+1}$.  
Given a subspace $x\subseteq PG(2k+1,\mathbb{F})$, this algorithm calculates the subspaces $c\in \mathcal{C}$ with $d(x,c)\leq t$.}
\label{alg:4a}
\begin{algorithmic}
\State Use Gaussian elimination to calculate the dimension of $x$. If $\dim(x)>k$, dualize. 
\State Use Algorithm \ref{alg:coord3} to calculate a basis $\{ v_{\bf i} \}$ of $W:=W(F(x,(k-t+\dim(x))/2,k))$. 
\State Define the linear parametrization $p(\alpha)=\mathbb{P}\left(\sum_{\bf i}\alpha_{\bf i}w_{\bf i}\right)$ of $\mathcal{W}=\mathbb{P}(W)$. 
\State Return $\mathcal{C} \cap \mathcal{W}$ as the solution of $E_{\mathcal{C}}(p(\alpha))$. 
\end{algorithmic}
\end{algorithm}

In this algorithm, $c$ is calculated as a solution of the system of equations consisting of the linear equations defining $\mathcal{U}$ together with the set of equations defining $\mathcal{V}$, applied to the linear parametrization $p(\alpha)$ of $\mathcal{W}(F(x,(k-t+\dim(x))/2,k))$. 

Note that Algorithm \ref{alg:0} in Section \ref{section:2} is a particular case of Algorithm \ref{alg:4a}. 
Therefore it is worth noting that in general, if $\dim(x)=k-1=\delta$, then, by Lemma \ref{lemma:3}, the projective variety $\mathcal{W}(F(x,\dim(x),k))$ is contained in the Pl\"ucker embedding of the Grassmannian $\mathcal{V}$.  
Also when $\delta<k-1$, if it is known that the sent subspace $c$ is contained in a given $(k+1)$-dimensional projective subspace $y$, then by restricting to points representing subspaces contained in $y$, by Lemma \ref{lemma:3}, the resulting subspace is contained in $\mathcal{V}$. 
In both these cases intersection with $\mathcal{V}$ is superfluous. 
This is interesting, since it allows for a solution similar to Algorithm \ref{alg:0}, where all calculations where made without Pl\"ucker coordinates. 
The first case corresponds to codes correcting one error. The second case could be useful for example over channels where the same codeword is sent more than once, by considering the smallest subspace containing all the received subspaces corresponding to the same sent codeword.

The second version of the geometric decoding algorithm uses the equations defining $\mathcal{W}(F(x,\delta,k))$ and $\mathcal{C}$. 
It is the version of the algorithm which most resembles the algorithm in \cite{Rosenthal2}, but it is substantially improved compared to \cite{Rosenthal2} by removing the high complexity caused by the need to solve a system of quadratic equations. Currently the best method for solving systems of non-linear equations requires the calculation of a Gr\"obner basis of the ideal of the polynomials defining the equations. This method can be regarded as a generalization to non-linear equations of Gaussian elimination of linear equations. 
The worst case complexity of the calculation of a Gr\"obner basis is exponential, but in the case of linear equations the algorithm reduces to Gaussian elimination, which has only cubic complexity. 

The improvement of Algorithm \ref{alg:4b} compared to the algorithm in \cite{Rosenthal2}  is due to the observation that the Gr\"obner basis of the polynomial ideal defining the code variety can be precalculated. 
Decoding can then be done by adding the linear equations defining $\mathcal{W}(F(x,\delta,k))$ to this Gr\"obner basis. 
The fact that all polynomials involved are homogeneous implies that the algorithm for adding these linear equations to the Gr\"obner basis reduces to Gaussian elimination. More precisely, we have the following result. 
\begin{lemma}
Let $K$ be a field and let $I$ be a homogeneous ideal of $K[X_1,\dots,X_n]$. For $i\in\mathbb{N}\cup\{0\}$, let $I_i\subseteq I$ be the ideal generated by the homogeneous elements of $I$ of degree $i$ and let $G_i$ be a Gr\"obner basis of $I_i$. Then $\bigcup_{i\in \mathbb{N}\cup\{0\}} G_i$ is a Gr\"obner basis of $I$.   
\end{lemma}
\begin{proof}
First, recall some terminology from the theory of Gr\"obner bases. If $J$ is an ideal of $K[X_1,\dots,X_n]$, then the initial ideal of $J$ (for a given monomial ordering $\prec$) is defined as the ideal generated by the initial terms of the polynomials in $J$ (for the monomial ordering $\prec$). We write $in_{\prec}(J)$.  
A set of polynomials $B$ in $J$ is a Gr\"obner basis of $J$ (for the monomial ordering $\prec$) if $in_{\prec}(B)=in_{\prec}(J)$, that is, if the leading terms of the elements in $B$ generates the ideal generated by the leading terms of the elements in $J$. 

Now, any leading term of a polynomial in $I$ is a leading term of some polynomial in some $I_i$. Therefore $in_{\prec}(I)\subseteq \sum_i in_{\prec}(I_i)$. Clearly $in_{\prec}(I)\supseteq \sum_i in_{\prec}(I_i)$, implying that $in_{\prec}(I)=\sum_i in_{\prec}(I_i)$.  
Therefore, if $H_i$ is a basis of $in_{\prec}(I_i)$ for all $i\in \mathbb{N}\cup\{0\}$, then $H=\bigcup_{i\in \mathbb{N}\cup\{0\}} H_i$ is a basis of $in_{\prec}(I)$. 

If $G_i$ is a Gr\"obner basis of $I_i$ for all $i\in \mathbb{N}\cup\{0\}$, then the initial terms of $G_i$ is a basis of $in_{\prec}(I_i)$ (that is, $in_{\prec}(G_i)=in_{\prec}(I_i)$). Consequently, the union over $i$ of the initial terms of $G_i$ forms a basis of $in_{\prec}(I)$, that is, $\bigcup_i G_i$ is a Gr\"obner basis of $I$. 
\end{proof}


\begin{corollary}
Let $G$ be a homogenous Gr\"obner basis of an ideal $I$. Let $L\subseteq G$ be the set of linear polynomials  of $G$ (possibly empty) and let $M$ be a set of linear polynomials. Let $H$ be a Gr\"obner basis of $\langle L\cup M\rangle$, then $G\cup H$ is a Gr\"obner basis of $\langle I\cup M\rangle$.
\end{corollary}
Note that the Gr\"obner basis of an ideal generated by linear polynomials can be calculated using Gaussian elimination. 

\begin{algorithm}[H]
\caption{Let $\mathcal{U}\subseteq PG\left(\binom{n+1}{k+1}-1,\mathbb{F}\right)$ be an algebraic variety defining a geometric subspace code $\mathcal{C}=\mathcal{V}\cap \mathcal{U}$ of constant (projective) dimension $k$ in $PG(n,\mathbb{F})$. Assume that $\mathcal{C}$ is given in terms of the reduced Gr\"obner basis of its polynomial ideal. 
Given a subspace $x\subseteq PG(n,\mathbb{F})$, this algorithm calculates the subspaces $c\in \mathcal{C}$ with $d(x,c)\leq t$. }
\label{alg:4b}
\begin{algorithmic} 
\State Use Algorithm \ref{alg:impl2} or Algorithm \ref{alg:impl} to calculate the equations defining $\mathcal{W}(F(x,(k-t+\dim(x))/2,k))$. 
\State Calculate $c=\mathcal{C} \cap \mathcal{W}(F(x,(k-t+\dim(x))/2,k))$ by solving the joint system of equations(, using Gaussian elimination). 
\end{algorithmic}
\end{algorithm}

\begin{theorem}
Algorithm \ref{alg:4b} decodes the subspace $x\subseteq PG(2k+1,\mathbb{F})$ to the subspaces $c\in \mathcal{C}$ with $d(x,c)\leq t$. If $\mathcal{C}$ is given in the form of a reduced Gr\"obner basis, then the order of complexity of Algorithm \ref{alg:4b} equals the order of complexity of solving a system of linear equations in $\binom{n+1}{k+1}$ variables. 
\end{theorem}
\begin{proof} 
The first step has complexity less than $O\left(\binom{n+1}{k+1}^2\right)$. 
In the last step the intersection $\mathcal{C}\cap \mathcal{W}_{\Omega}$ is calculated as the solutions to the joint system of polynomial equations. 
The polynomials defining $\mathcal{C}$ are the quadratic polynomials defining $\mathcal{V}$ and the polynomials defining the variety $\mathcal{U}$. They are all homogeneous polynomials, and therefore any Gr\"obner basis also consist of homogeneous polynomials. 
The algorithm for adding the linear equations defining $\mathcal{V}$ to this basis reduces to Gaussian elimination. 
Therefore the overall complexity of Algorithm \ref{alg:4b} is the complexity of Gaussian elimination of a system of linear equations in $\binom{n+1}{k+1}$ variables, that is, smaller than $O\left(\binom{n+1}{k+1}^3\right)$.    
\end{proof}

The third version of the algorithm uses parametrizations of both $\mathcal{W}(F(x,\delta,k))$ and $\mathcal{C}$. It is as efficient as Algorithm \ref{alg:4b}, but this efficiency requires a good parametrization of $\mathcal{C}$. 
\begin{algorithm}[H]
\caption{Let $\mathcal{U}\subseteq PG\left(\binom{n+1}{k+1}-1,\mathbb{F}\right)$ be an algebraic variety defining a subspace code $\mathcal{C}=\mathcal{V}\cap \mathcal{U}$ of constant (projective) dimension $k$ in $PG(n,\mathbb{F})$. Assume that $\mathcal{C}$ is given in terms of a local polynomial parametrization $q(\beta)$. Assume that the  polynomials of $q(\beta)$ form a reduced Gr\"obner basis.  
Given a subspace $x\subseteq PG(n,\mathbb{F})$, this algorithm calculates the subspaces $c\in \mathcal{C}$ with $d(x,c)\leq t$. }
\label{alg:4c}
\begin{algorithmic} 
\State Use Algorithm \ref{alg:coord3} to calculate a parametrization $p(\alpha)$ of $\mathcal{W}(F(x,(k-t+\dim(x))/2,k))$. 
\State Solve the system of equations obtained by setting $p(\alpha)=q(\beta)$ (or $q(\beta) - p(\alpha)=0$). Note that if $q(\beta)$ is written as a Gr\"obner basis, only Gaussian elimination of $p(\alpha)$ (with respect to $q(\beta)$) is needed to solve this system. 
\State Return the solution: $\mathcal{C} \cap \mathcal{W}(F(x,(k-t+\dim(x))/2,k))$. 
\end{algorithmic}
\end{algorithm}


Algorithm \ref{alg:4a}, Algorithm \ref{alg:4b} and Algorithm \ref{alg:4c} can be applied to any subspace code which is an algebraic variety $\mathcal{U}\cap\mathcal{V}$ in the Pl\"ucker embedding of the Grassmannian, defined by a set of polynomials, that is, to any geometric subspace code.

\section{Decoding Desarguesian $t$-spreads}
\label{section:4}
A $t$-spread in $G_{\mathbb{F}}(t+1,2t+2)$ can correct $t$ errors. In particular, the spread of lines in $G_{\mathbb{F}}(2,4)$ from Section \ref{section:2} can correct one error. 
Here we will see how the decoding algorithm presented there generalizes to Desarguesian spreads in higher dimensions. 
We will use the decoding algorithm described in Section \ref{sec:decoder}. 
The algorithm is designed to decode any subspace code whose Pl\"ucker coordinates is an algebraic variety of the Pl\"ucker embedding of the Grassmannian. 
It is therefore enough to show that a Desarguesian $t$-spread is an algebraic variety in the Pl\"ucker embedding of the Grassmannian. 

In Section \ref{section:2}, we described how a Desarguesian line spread  is represented in the Pl\"ucker embedding of the Grassmannian $G_{\mathbb{F}}(2,4)$ as the complete intersection with a linear subspace. This generalizes to Desarguesian $t$-spreads in $PG(rt-1,\mathbb{F})$. 
The Desarguesian spreads are the spreads isomorphic to the classical spreads, see \cite{Segre}. 
\begin{theorem}
\label{lemma:spreadrep}
\cite{Segre, Lunardon,Lunardon2}
The Pl\"ucker coordinates of the subspaces of a Desarguesian $t$-spread in $PG(rt-1,\mathbb{F})$ are the points in the complete intersection of the Pl\"ucker embedding $\mathcal{V}$ of the Grassmannian and a projective subspace $\mathcal{U}$ of dimension $r^t-1$.  
\end{theorem}

A \emph{cap} of a projective space is a set of points such that no three of them are collinear. The set of points of $\mathcal{V}\cap\mathcal{U}$ forms a cap of $\mathcal{U}\sim PG(r^t-1,\mathbb{F})$ \cite{Lunardon}. 
Moreover, if $r=2$, then any $t+1$ points of  $\mathcal{V}\cap\mathcal{U}$ are in general position \cite{Pepe}.  

Algorithm \ref{alg:4a} from Section \ref{sec:decoder} is designed for subspace codes in $G_{\mathbb{F}}(k+1,2k+2)$. 
The elements of such a code are situated in the middle of the subspace lattice of $PG(2k+1,\mathbb{F})$. This implies that error-correction dualizes nicely. 
The $t$-spreads are the largest codes with the maximum minimum distance of all codes in $G_{\mathbb{F}}(t+1,2t+2)$. Therefore $t$-spreads in $PG(2t+1,\mathbb{F})$ are good codes.  
 
Lemma \ref{lemma:spreadrep} gives a representation of Desarguesian $t$-spreads in the Pl\"ucker embedding of the Grassmannian as an algebraic variety, the section of the Pl\"ucker embedding of the Grassmannian and a projective subspace $\mathcal{U}\subseteq PG\left(\binom{2t+2}{t+1}-1,\mathbb{F}\right)$. 
This representation makes it possible to decode Desarguesian $t$-spread codes using Algorithm \ref{alg:4a}, Algorithm \ref{alg:4b} and Algorithm \ref{alg:4c}. 
Note that such spreads have relatively small decoding complexity with Algorithm \ref{alg:4a}, since the equations defining them as an algebraic variety in the Grassmannian are linear. In Algorithm \ref{alg:4b} and Algorithm \ref{alg:4c}, the equations and the parametrization of the code are precalculated and presolved, respectively, so here the linearity of the code is less important. 
%
%
%
For $t=1$, Algorithm \ref{alg:4a} applied to spread codes becomes the algorithm for Desarguesian line spread codes in $PG(3,\mathbb{F})$ from Section \ref{section:2}. 

Another family of subspace codes which are defined as the intersection of the Pl\"ucker embedding of the Grassmannian with a projective subspace (that is, a linear variety) are the lifted Gabidulin codes. Note that the Desarguesian spread codes are actually an example of lifted Gabidulin codes. 



\section*{Conclusions}

The subject of this article was a decoding algorithm of geometric subspace codes, that is, constant dimension subspace codes of vector space dimension $k+1$ in a vector space of dimension $n+1$, which are defined in $PG\left(\binom{n+1}{k+1}-1,\mathbb{F}\right)$ as the points in the intersection of an algebraic variety and the Pl\"ucker embedding of the Grassmannian $G_{\mathbb{F}}(k+1,n+1)$. 
The complexity of some versions of this algorithm was shown to be the complexity of solving a system of linear equations in $\binom{n+1}{k+1}$ variables. 
This was due to the observation that the equations defining the code inside the Pl\"ucker embedding of the Grassmannian can be solved once and for all in the construction of the code. 
Therefore the complexity of decoding a received subspace $x$ is reduced to solving the linear equations defining the Schubert variety consisting of the set of subspaces in $G_{\mathbb{F}}(k+1,n+1)$ which are located within a certain distance from $x$. 

It was first thought that the complexity of decoding geometric subspace codes depends highly on the degree of the polynomials defining the code in the Grassmannian. 
However the results presented in this paper show that this is not the case. 
The decoding algorithm can be designed to have the order of complexity of solving a system of linear equations in $\binom{n+1}{k+1}$ variables. This is Algorithm \ref{alg:4b} and \ref{alg:4c}. 
For codes which can correct a single error, the decoding complexity can be reduced so that it has the order of complexity of solving a system of linear equations in $n+1$ variables, since it is then not necessary to use the Pl\"ucker coordinates. For this purpose we have used Algorithm \ref{alg:4a}. Section \ref{section:2} gives a thorough example on how to realize this idea in practice and shows an interesting link to the Klein correspondence.  

We have applied the geometric decoding algorithm to spread codes, giving an example of how finite geometry is an important tool in the construction of geometric subspace codes. 
Desarguesian spread codes are defined as the intersection of a linear variety and the Pl\"ucker embedding of the Grassmannian. 
They make an example from an entire family of subspace codes defined in the same way: the lifted Gabidulin codes, or more generally, lifted linear rank-metric codes (see also \cite{Rosenthal2}).

\section*{Acknowledgments}
The author would like to thank Axel Hultman for many useful discussions during the elaboration of this article. 
Partial financial support from the Spanish MEC project ICWT (TIN2012-32757) is acknowledged.

\end{document}